\documentclass[a4paper,12pt]{amsart}

\usepackage{xcolor}

\usepackage{tikz}
\usetikzlibrary{shapes.geometric,positioning}

\usepackage[margin=3cm]{geometry}
\newtheorem{thm}{Theorem}[section]

\newtheorem{lem}[thm]{Lemma}
\newtheorem{cor}[thm]{Corollary}

\theoremstyle{definition}
\newtheorem{defn}[thm]{Definition}

\newcommand{\lcut}{{\mathcal{LCUT}}}

\newcommand{\TBN}{{T\negthinspace B\negthinspace N}}

\title{The space of tree-based phylogenetic networks}
\author{Mareike Fischer}
\author{Andrew Francis}
\address{Institute of Mathematics and Computer Science, University of Greifswald, Germany}
\email{email@mareikefischer.de}
\address{Centre for Research in Mathematics and Data Science, Western Sydney University, Australia}
\email{a.francis@westernsydney.edu.au}

\date{\today}
\begin{document}

\begin{abstract}
Phylogenetic networks are generalizations of phylogenetic trees that allow the representation of reticulation events such as horizontal gene transfer or hybridization, and can also represent uncertainty in inference.  A subclass of these,
\emph{tree-based} phylogenetic networks, have been introduced to capture the extent to which reticulate evolution nevertheless broadly follows tree-like patterns.  Several important operations that change a general phylogenetic network have been developed in recent years, and are important for allowing algorithms to move around spaces of networks; a vital ingredient in finding an optimal network given some biological data. A key such operation is the Nearest Neighbor Interchange, or NNI. While it is already known that the space of unrooted phylogenetic networks is connected under NNI, it has been unclear whether this also holds for the subspace of tree-based networks. In this paper we show that the space of unrooted tree-based phylogenetic networks is indeed connected under the NNI operation. We do so by explicitly showing how to get from one such network to another one without losing tree-basedness along the way. Moreover, we introduce some new concepts, for instance ``shoat networks'', and derive some interesting aspects concerning tree-basedness. Last, we use our results to derive an upper bound on the size of the space of tree-based networks.  
\end{abstract}
\maketitle

\section{Introduction} 
Phylogenetic networks have become widely studied structures in the mathematics of evolution, because they capture a realistic range of evolutionary events beyond speciation, which itself is encapsulated elegantly via phylogenetic trees.  In particular, phylogenetic networks are able to represent explicit events such as horizontal gene transfer and hybridization, as well as the presence of uncertainty. Phylogenetic networks, like trees, have also been widely studied using geometric approaches that consider the set of networks as a space in which one may move between the objects by operations that change a feature of the graph.  Such operations, prominent examples of which use the \textit{nearest neighbor interchange} (NNI), \textit{subtree prune and regraft} (SPR), and \emph{tree bisection and reconnection} (TBR), generalized from trees, are valuable in computational applications because it is often necessary to search the space of networks to find one that is optimal with respect to some specific criterion~\cite{robinson1971comparison,allen2001subtree,huber2016transforming}.  Furthermore, such operations may be used to define proximity measures to tree-based networks, extending the notion of tree-basedness to the full class of phylogenetic networks~\cite{fischer2019how}.

At the same time, motivated by various biological and mathematical considerations, various classes of phylogenetic network have been intensively studied.  The focus of this paper is the class of \textit{tree-based networks}, that are essentially ``trees with additional edges'' that pass between the tree edges~\cite{francis2015phylogenetic,francis2018tree,jetten2016nonbinary,hendriksen2018tree}.  
While the space of phylogenetic networks of a given ``tier'' (defined below) is connected under  NNI~\cite{huber2016transforming}, 
we show in this paper that the  space of tree-based unrooted phylogenetic networks is also connected under NNI moves. 
That is, it forms a connected subspace of the space of phylogenetic networks.  Likewise, the diameter of the space of unrooted phylogenetic networks has been bounded in~\cite{francis2018bounds}; here we show that for a given tier, the diameter of the space of unrooted tree-based phylogenetic networks is of order $\mathcal O(n^2)$, where $n$ is the number of leaves.

We begin by setting out the necessary definitions, including a new class of (tree-based) phylogenetic networks that we call \textit{shoat} networks.   In Section~\ref{s:NNI.among.TBNs} we describe circumstances in which an NNI move on a tree-based network will produce another tree-based network (not all do), and show that it is possible to move between various classes of tree-based network using NNI moves, while staying tree-based.  Our main result, the connectedness of the space of tree-based networks, is proved in Section~\ref{s:TBN.connectedness}; we finish with a discussion and further questions in Section~\ref{s:discussion}.

While this manuscript was in preparation, a paper containing a statement similar to one of our main results appeared on the arXiv~\cite[Theorem~4.10]{janssen2019rearrangement}. The result in that paper relates to unrooted networks that differ from ours in that they permit parallel edges, i.e. they are multigraphs. However, in the context of phylogenetic networks, it is more common (and arguably biologically more natural) to consider simple graphs~\cite{pons2018tree,fischer2018nonbinary,francis2018tree,steel2016phylogeny,gambette2012quartets}. 
In contrast, our work takes a different approach to the question, and it is not clear whether the approach in~\cite{janssen2019rearrangement} can be adapted to deal with networks as we define them here, i.e. without parallel edges.

\section{Definitions and background}\label{s:definitions}

A {binary} unrooted phylogenetic network on a set $X$ (typically a set of species or taxa) is a connected simple graph whose vertices are degree 1 or 3, and whose degree 1 vertices (leaves) are {bijectively} labelled by the elements of $X$.  Some of our results extend to the non-binary case, in which non-leaf vertices may have degree greater than 3, and we will remark on that where it arises. In the following, whenever there is no ambiguity, we use the term network to refer to an unrooted phylogenetic network.
Note that the special case of an acyclic phylogenetic network is called a phylogenetic tree.

Throughout this manuscript, we will  assume that $|X|\ge 2$, and that $N$ is ``proper''. A proper network is one for which all components obtained by removing a cut edge or cut vertex contain at least one element of $X$.

In the following, we denote by $V^1(N)$ the set of degree $1$ vertices in $N$, i.e. the set of leaves $V^1(N)=X$, and $\mathring{V}=V\setminus V^1$ denotes the set of inner vertices of $N$.
If $k$ is minimal such that the deletion of ${k}$ edges of ${N}$ would turn ${N}$ into a tree (i.e. a connected acyclic graph), we say that ${N}$ has \emph{tier} ${k}$. Note that the tier does not depend on $N$ being a phylogenetic network -- in fact, the tier of a connected graph can be defined analogously, and for technical reasons, we need this later on in this manuscript. We denote the set of tier-$k$ phylogenetic networks on $n$ leaves by $N(n,k)$.

The \textit{triangle operation}, introduced in~\cite{huber2016transforming},  allows the replacement of a vertex with a triangle (a cycle of length three) to go up a tier, and the reverse to go down a tier.

More formally, the ``blow-up'' triangle operation $\Delta^+$ (that raises the tier of a network by 1), replaces an internal, degree 3 vertex $v$, and its incident edges $\{v,w_1\}$, $\{v,w_2\}$, $\{v,w_3\}$, by three vertices $v_1,v_2,v_3$ and six new edges: $\{v_1,w_1\}$, $\{v_2,w_2\}$, $\{v_3,w_3\}$, $\{v_1,v_2\}$, $\{v_2,v_3\}$, $\{v_3,v_1\}$.  Inversely, the ``collapse'' triangle operation $\Delta^-$ (lowering the tier by 1), replaces a triangle (a 3-cycle) and the edges outside the cycle but incident to its vertices  by a single vertex and three incident edges.  That is, given a 3-cycle of vert  ices $\{v_1,v_2,v_3\}$ and edges $\{v_1,v_2\}$, $\{v_2,v_3\}$, $\{v_3,v_1\}$, and three incident edges $\{v_1,w_1\}$, $\{v_2,w_2\}$, $\{v_3,w_3\}$, replace these by a single vertex $v$ and three edges $\{v,w_1\}$, $\{v,w_2\}$, $\{v,w_3\}$.  We denote these operations $\Delta^+(N,v)$ and $\Delta^-(N,\{v_1,v_2,v_3\})$, for $v$ an internal vertex and $\{v_1,v_2,v_3\}$ a 3-cycle in the network $N$.

A {\it blob} of a network (or, more generally, of a graph) is a maximal connected subgraph that has no cut edge (if such a blob consists of only one vertex, it is called {\it trivial}). Note that in a binary phylogenetic network, blobs cannot contain any cut vertices (as all cut vertices in a binary network are incident to a cut edge~\cite[Lemma 8]{fischer2018nonbinary}).
A phylogenetic network is called {\it simple} if it contains at most one non-trivial blob.

A \emph{support tree} $T$ of a network $N$ is a spanning tree of $N$ satisfying $V^1(T)=V^1(N)=X$, that is, whose leaf set coincides with the leaf set $X$ of $N$.  If $N$ contains such a support tree $T$, it is called \emph{tree-based}. Note that a support tree $T$ of $N$ is not necessarily a phylogenetic tree as it may contain degree-2 vertices. The space of tree-based networks in tier $k$ on $n$ leaves is denoted $\TBN(n,k)$.

Another concept that we need in the following, recently introduced in~\cite{fischer2018classes}, is the \emph{leaf cut graph} $\lcut(N)$ of a proper network $N\in N(n,k)$, with $|V(N)|\ge 3$, 
which is the graph $G$ obtained from $N$ by deleting all leaves and their incident edges. Note that this may result in some vertices of degree 2 and -- e.g. if $N$ is a tree -- even new leaves not labelled by $X$, which we do \emph{not} remove.  

Our main results require the notions of \emph{lineal} tree-based networks and \emph{shoat} networks.

A tree-based network is called \emph{lineal} if it has a support tree consisting of a single path $p$ between two leaves, with paths from $p$ to other leaves of length one.  For instance, a tree-based network that is \textit{pseudo-Hamiltonian} (in the sense of~\cite{francis2018bounds}: a network whose $\mathcal{LCUT}(N)$ 
graph~\cite{fischer2019how} has a Hamiltonian cycle, cf. also ~\cite{fischer2018classes}) will be lineal if there are two vertices that are connected to leaves, and that are adjacent on the Hamiltonian cycle. 

The second new family of graphs that we define are called \emph{shoat networks}\footnote{Because of their close resemblance to the juvenile boars that frequent the streets of northern Germany.}.   These are a subclass of the pseudo-Hamiltonian, tree-based phylogenetic networks, defined below.

\begin{defn}\label{d:shoat}
A {binary} {tier $k\ge 1$} phylogenetic network $N\in\TBN(n,k)$ is a \emph{shoat network} if it has a Hamiltonian 
$\mathcal{LCUT}(N)$ graph with the properties that:
\begin{enumerate}
    \item there are two leaves $x,y\in X$ whose adjacent interior vertices $a$ and $b$ are adjacent to each other, and such that there is a path $p$ from $x$ to $y$ that visits all interior vertices of $N$; 
    \item the interior vertices of $N$ are partitioned into sets $V_L$, $V_M$, $V_R$ and $\{a,b\}$, with properties:
    \begin{itemize}
        \item the path $p$ consists (in order) of $x,a$, elements of $V_L$, elements of $V_M$, elements of $V_R,b,y$;
        \item $|V_L|,|V_R|= k-1$, $|V_M|= n-2$; 
        \item each vertex in $V_L$ is adjacent to its neighbouring vertices on $p$ and one vertex in $V_R$;
        \item each vertex in $V_R$ is adjacent to its neighbouring vertices on $p$ and one vertex in $V_L$;
        \item each vertex in $V_M$ is adjacent its neighbouring vertices on $p$, and a leaf in $X\setminus\{x,y\}$.
    \end{itemize}
\end{enumerate}   
\end{defn}
Note, shoat networks can be defined in the nonbinary setting by allowing $|V_L|,|V_R|\le k-1$ and $|V_M|\le n-2$, and allowing vertices in $V_L$ and $V_R$ to also be adjacent to one or more leaves.

An illustration of this definition is given in Figure~\ref{f:shoat}.

\begin{figure}[ht]
    \centering
\includegraphics[width=.6\textwidth]{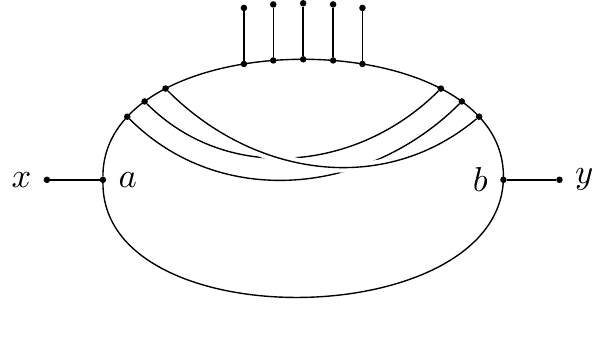}
    \caption{An example of a nonbinary shoat network, with $k=4$ and $n=7$. Here $|V_L|=|V_R|=3$, $|V_M|=5$. }
    \label{f:shoat}
\end{figure}

The \emph{nearest neighbor interchange (NNI)} is a replacement of a path in the network with an alternative path, and was defined for unrooted phylogenetic networks in~\cite{huber2016transforming}:

\begin{defn}[NNI]
Let $N$ be a phylogenetic network in which $(a,b,c,d)$ is a path for which neither $\{a,c\}$ nor $\{b,d\}$ is an edge.  The NNI operation on this path replaces it with the path $(a,c,b,d)$: the edges $\{a,b\}$ and $\{c,d\}$ are deleted, and edges $\{a,c\}$ and $\{b,d\}$ are added.
\end{defn}

\section{NNI moves among tree-based networks}\label{s:NNI.among.TBNs}

The main result of this paper, Theorem~\ref{t:TBN.NNI.connected}, proves that tree-based networks are connected under NNI moves.  In this section,  we prove some preliminary results about the effect of NNI moves on tree-based networks, {which will all be needed for the proof of this main theorem,} including showing that some different subclasses of tree-based networks are connected. We begin by proving that certain NNI moves preserve tree-basedness.

\begin{lem}\label{l:tree-based-preserving-moves}
If $N$ is a tree-based network with support tree $T$, then an NNI move on the path $(a,b,c,d)$ will produce another tree-based network if either:
\begin{enumerate}
    \item the edges $\{a,b\},\{b,c\}$, and $\{c,d\}$ are all in $T$;
    \item $\{b,c\}$ is in $T$ but $\{a,b\}$ and $\{c,d\}$ are not; or
    \item $\{b,c\}$, $\{c,d\}$, and $\{c,d'\}$ are in $T$ for some other vertex $d'$, but $\{a,b\}$ is not.
\end{enumerate}
\end{lem}
\begin{proof}\mbox{}
\begin{enumerate}
    \item Suppose $T$ is a support tree for the network $N$.  If all edges in the path $(a,b,c,d)$ are in $T$, shown as bold dashed and solid lines in Figure~\ref{f:NNI-tree-base-path}, then the effect of the NNI produces a new support tree $\tilde T$ that has edges $\{a,c\}$ and $\{b,d\}$ instead of $\{a,b\}$ and $\{c,d\}$, as shown in the right side of the Figure.  It is clear that $\tilde T$ is still a support tree for $N$ because it still covers all vertices in $N$; it has no additional leaves beyond those of $T$; and it has not generated any cycles (there are still unique paths between all vertices $a,b,c,d$ in $\tilde T$, as there were in $T$).

\begin{figure}[ht]
\includegraphics[width=12cm]{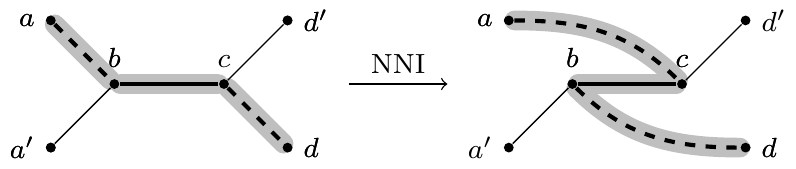}
\caption{An NNI move, such as this one on the path $a,b,c,d$, can preserve tree-based-ness in certain contexts (Lemma~\ref{l:tree-based-preserving-moves}). }
\label{f:NNI-tree-base-path}
\end{figure}

\item If on the other hand neither $\{a,b\}$ nor $\{c,d\}$ are in the support tree $T$, but $\{b,c\}$ is, then there must be another edge incident to $b$ that is in $T$ (clearly, $b$ is not a leaf in $N$ because it is contained in the path $(a,b,c,d)$, i.e. it has degree $>1$ in $N$ and thus also cannot be a leaf in $T$).  Likewise there is an edge in $T$ incident to $c$ other than $\{b,c\}$, for the same reasons.  Without loss of generality suppose the edges $\{a',b\}$ and $\{c,d'\}$ are in $T$, shown in Figure~\ref{f:NNI-tree-base-path} as solid lines.  Then the NNI on the path $(a,b,c,d)$ has no effect at all on the spanning tree $T$, but simply shifts around the attachment edges that pass between vertices of the support tree $T$.  Thus, $T$ remains a support tree for $N$.

\item Finally, if $(b,c,d)$ is a path in $T$ as described, and $\{c,d'\}$ another edge in $T$, then since $b$ is not a leaf in $N$, there is another edge in $T$ incident to it other than $\{a,b\}$ (because $\{a,b\}$ is not contained in $T$ by assumption in this case, but $b$ cannot be a leaf of $T$ as it is not a leaf in $N$, either): say without loss of generality that it is $\{a',b\}$.  Then the effect of the NNI move on $(a,b,c,d)$ is to shift the points that $a$ and $d$ connect to the path $(a',b,c,d')$ in $T$ from $b$ to $c$ and from $c$ to $b$ respectively.  These shifts do not cause any vertex to become uncovered by the tree, and cannot create cycles.
\end{enumerate}
Thus, in all cases, the NNI move sending $(a,b,c,d)$ to $(a,c,b,d)$ changes the tree-based network $N$ to another tree-based network.
\end{proof}

Note, other scenarios may or may not preserve tree-basedness.  For instance, if the path $(a,b,c,d')$ in Figure~\ref{f:NNI-tree-base-path} was in a support tree for $N$ but the edges $\{a',b\}$ and $\{c,d\}$ were not, then the NNI move shown in the figure could turn $b$ into a leaf of the spanning tree, which would then contain edges $\{a,c\}$, $\{c,d'\}$, and $\{b,c\}$.  Whether it did in fact or not would depend on other features of the network outside the local frame of this NNI move. 

Examples of NNI moves that make a tree-based network not tree-based are plentiful.
For instance, there are exactly two level-5 non-tree-based binary networks~\cite{fischer2018nonbinary}, and so there are many NNI moves from each of them that lead to tree-based networks.  Consequently each of those NNI moves in reverse makes a tree-based network non-tree-based.

In the remainder of this section we prove several lemmas that show that a tree-based network can be transformed to any shoat network by NNI moves that stay within $\TBN(n,k)$. It is important to note that all these NNI moves can of course be reverted -- i.e. the same arguments can be used to go back from said shoat network to the original tree-based network while staying within $\TBN(n,k)$.

\begin{lem}\label{l:simple.to.lineal}
A network in $\TBN(n,k)$ can be transformed to a lineal tree-based network using only NNI moves  within $\TBN(n,k)$.
\end{lem}
\begin{proof}
Let $N$ be a tree-based network. Of all support trees of $N$, choose a tree $T$ which has the longest maximal length path $p$. Let this maximal length path consist of vertices $v_1,\dots,v_t$ in $T$. Suppose there is a vertex $v_i$ on the path for which there is a path in $T\setminus p$ to a leaf of length greater than 1 (note $i\neq 1$ or $t$, by maximality of $p$, as this implies that both $v_1$ and $v_t$ must be leaves and thus have no neighbours outside the path).  As long as $N$ is non-lineal, such a vertex $v_i$ must exist.  We argue that there is an NNI move that extends the length of the maximal path, and reduces the length of the path from $v_i$ to a leaf, as follows.    

Let $w_1$ and $w_2$ be the first two vertices along the path of length greater than 1 from $v_i$ in $T$, with $w_1$ adjacent to $v_i$. Note that this necessarily implies that $w_1$ cannot be adjacent to either $v_{i-1}$ or $v_{i+1}$ in $N$, because otherwise, the path $v_1,\dots,v_t$ could be extended by 
replacing $(v_{i-1},v_i,v_{i+1})$ by $(v_{i-1},w_1,v_i,v_{i+1})$ (or by $(v_{i-1},v_i,w_1,v_{i+1})$, respectively). This would turn $T$ into a support tree $T'$ of $N$ with a longer maximal path, contradicting the choice of $T$. 

So as $w_1$ is not adjacent to either $v_{i-1}$ or $v_{i+1}$ in $N$, we can perform a legal NNI move on the path $(v_{i-1},v_i,w_1,w_2)$, to create the path $(v_{i-1},w_1,v_i,w_2)$, as shown in Figure~\ref{f:extend.max.path}. Then the modified tree $T$ has a new maximal path $(v_1,\dots,v_{i-1},w_1,v_i,\dots,v_t)$, and the length of the other path from $v_i$ to a leaf has been reduced in length by 1. Furthermore, it is still tree-based, by Lemma~\ref{l:tree-based-preserving-moves}(i).
\begin{figure}[ht]
\includegraphics[width=.7\textwidth]{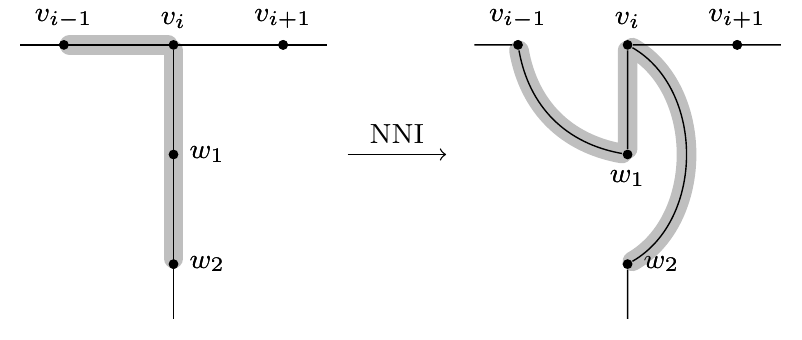}
\caption{An NNI move that incorporates an edge from $v_i$ into a maximal path $(\dots,v_{i-1},v_i,v_{i+1},\dots)$ within the support tree $T$ of a tree-based network. Edges in $N$ that are not in $T$ are not shown.}
\label{f:extend.max.path}
\end{figure}

This process can be continued until the support tree $T$ of $N$ takes the form of a single long path with edges from it connecting directly to leaves of $N$ --- that is, until the resulting tree-based network is lineal. 
\end{proof}

\begin{lem}\label{l:lineal.to.pseudo.H} 
A lineal phylogenetic network in $\TBN(n,k)$, for $k\ge 1$, can be transformed into a shoat network in $\TBN(n,k)$  by NNI moves within $\TBN(n,k)$.
\end{lem}

\begin{proof}
Let $N$ be a lineal phylogenetic network in $\TBN(n,k)$ with $k\ge 1$.  

Choose a maximal length path $p$ in $N$. Because $N$ is lineal, this {induces} also a support tree $T$ for $N$. Label the leaves at each end of the path $x$ and $y$, and the vertices along $p$ by $x=v_1,v_2,\dots,v_m=y$.
Because $N$ is lineal, the edges connected to the other vertices along the path are either incident to leaves (if $n>2$), or other vertices along the path.

Let 
\begin{align*}
    L&=\{v_i\mid \{v_i,v_j\}\in E(N), i+1<j\}\\
    R&=\{v_j'\mid \{v_i',v_j'\}\in E(N), i+1<j\},
\end{align*}
so that $L$ (resp. $R$) is the set of vertices at the left (resp. right) end of edges between internal vertices that are not on the path $p$ (we assume without loss of generality that `left' refers to vertices closer to $x$ and `right' refers to vertices closer to $y$).
These sets are disjoint, since each vertex has degree 3, and non-empty, since $k\ge 1$.
There may also be vertices in $L$ to the right of some (but not all) vertices in $R$; let 
\[\Theta:=\{(v_i,v_j)\in L\times R\mid i>j\}\]
be the ``overlap set'': the pairs of vertices $(v_i,v_j)$ for which $v_i$ is an element of $L$ but is nevertheless to the right of $v_j\in R$, as in Figure~\ref{f:lineal.overlap}.
Define the \textit{extent} of the overlap to be the non-negative integer:
\[||\Theta||:=\sum_{(v_i,v_j)\in\Theta}(i-j).\]

\begin{figure}[ht]
    \centering
    \includegraphics[width=.8\textwidth]{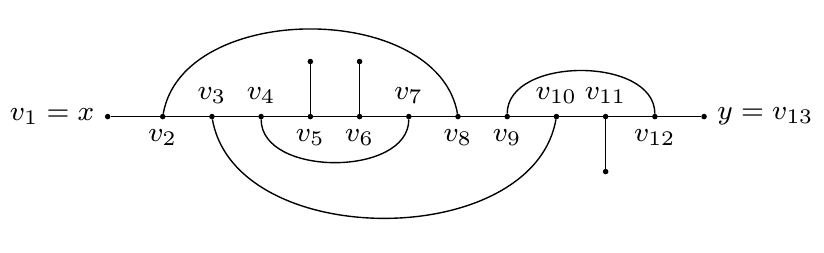}
    \caption{A lineal network with $L=\{v_2,v_3,v_4,v_9\}$, $R=\{v_7,v_8,v_{10},v_{12}\}$. This network has overlap set $\Theta=\{(v_9,v_7),(v_9,v_8)\}$ and $||\Theta||=2+1=3$.}
    \label{f:lineal.overlap}
\end{figure}

We will describe NNI moves that change the network to strictly decrease $||\Theta||$.

Suppose $\Theta\neq\emptyset$, and let $i$ be minimal such that  $(v_i,v_j)\in\Theta$.  
Because $i$ is minimal, the vertex $v_{i-1}$ to the left of $v_i$ in $p$ must either be (i) in $R$ or (ii) connected to a leaf.
We now consider these two cases.

(i) Suppose $v_{i-1}\in R$, and note that this means $i>4$ since there must be a vertex to the left of $v_{i-1}$ that is connected to it by an edge not in $p$, $v_1=x$ is a leaf, and $N$ has no parallel edges.
We perform the NNI move
\[ v_{i-2},v_{i-1},v_i,v_{i+1}\to v_{i-2},v_i,v_{i-1},v_{i+1},\]
shown in Figure~\ref{f:Theta.reduction.i}.
The result is still tree-based by Lemma~\ref{l:tree-based-preserving-moves}(i).
\begin{figure}[ht]
    \centering
    \includegraphics{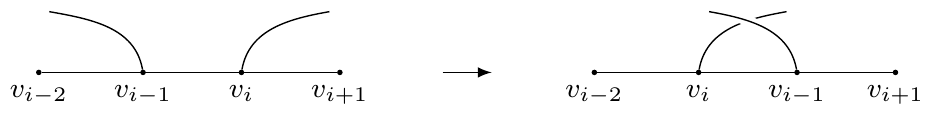}
    \caption{An NNI move reducing $||\Theta||$ by moving a vertex in $L$ to the left of a vertex in $R$. }
    \label{f:Theta.reduction.i}
\end{figure}

Now consider the effect on $\Theta$ of this NNI move.
Before the move, the elements of $\Theta$ that involve $v_{i-1}$ or $v_i$ are precisely the elements of the disjoint union:
\[\{(v_i,v_{i-1})\}\mathbin{\dot{\cup}}\{(v_j,v_{i-1})\mid \forall j>i, v_j\in L\}\mathbin{\dot{\cup}}\{(v_i,v_{k})\mid \forall k<i-1, v_k\in R\}.\]
After the NNI move, the elements of $\Theta$ that involve $v_{i-1}$ or $v_i$ are precisely
\[\{(v_j,v_{i})\mid \forall j>i, v_j\in L\}\mathbin{\dot{\cup}}\{(v_{i-1},v_{k})\mid \forall k<i-1, v_k\in R\}.\]
Note that elements that do not involve $v_{i-1}$ or $v_i$ are unaffected by the move, and that the cardinalities of the corresponding sets have not changed, but each corresponding term in $||\Theta||$ has decreased by 1.  Therefore the value of $||\Theta||$ has been decreased by the move by 
\[1+|\{v_j\in L\mid j>i\}| + |\{v_k\in R\mid k<i-1\}|.\]
In particular,
$||\Theta||$ has strictly decreased as a result of the move.

(ii) If $v_{i-1}$ is not in $R$ but connects to a leaf $z$ (so that in particular, $n\ge 3$), there is at least one vertex $v_j$ further left of $v_{i-1}$ that is in $R$. 
Perform the NNI move shown in Figure~\ref{f:Theta.reduction.ii}:
\[v_{i-2},v_{i-1},v_i,v_{i+1}\to v_{i-2},v_i,v_{i-1},v_{i+1}.\]
\begin{figure}[ht]
    \centering
    \includegraphics{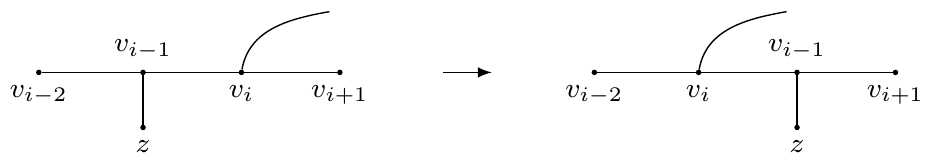}
    \caption{An NNI move reducing $||\Theta||$ by moving a vertex in $L$ to the left of a vertex attached to a leaf.} 
    \label{f:Theta.reduction.ii}
\end{figure}
As before, this still results in a tree-based network by Lemma~\ref{l:tree-based-preserving-moves}(i).  After the operation and relabelling vertices along the path, $i$ is no longer minimal such that $(v_i,v_j)\in\Theta$, because $(v_{i-1},v_j)$ is now in $\Theta$. 

The effect on $\Theta$ of this move is that an ordered pair $(v_i,v_j)\in L\times R$ with $j<i$ is replaced in $\Theta$ by the pair $(v_{i-1},v_j)$. As before, this has the effect of strictly reducing $||\Theta||$, as required.

Since such moves are always possible while $||\Theta||>0$, and they each strictly reduce $||\Theta||$, it follows that successive application of these moves can reduce $||\Theta||$ to zero.  At that point we must have $\Theta=\emptyset$.

Having reduced the size of the overlap set $||\Theta||$ to zero, we need to move any leaves (apart from $x$ and $y$) to the middle.

Suppose there are leaves attached to vertices $v_j$ in $p$, for $j\neq 1,m$, other than $x=v_1$ and $y=v_m$. We now want to bring those leaves to the middle of the path, between the vertices in $L$ and those in $R$.

Suppose a leaf {$z\neq x$} is attached to a vertex $v_{j-1}$ which has a vertex $v_{j}\in L$, to its right.
Perform the NNI move 
\[v_{j-2},v_{j-1},v_{j},v_{j+1}\ \to\ v_{j-2},v_{j},v_{j-1},v_{j+1},\]
which moves the attachment of the leaf $z$ one position towards the right along the path $p$ (the same move shown in Figure~\ref{f:Theta.reduction.ii}).  The new network is still tree-based, by Lemma~\ref{l:tree-based-preserving-moves}(i).  As long as there is a leaf attached to the left of a vertex from $L$, this can be repeated until there are none left.  Similarly, the reverse can be performed for leaves that are attached to a vertex that has a vertex in $R$ to their left, until there are no such vertices.

The result is that the leaves (apart from $x$ and $y$) are attached to vertices in $p$ with all vertices in $L$ to their left, and all vertices in $R$ to their right.

It is now elementary to make the network shoat.  Take the leftmost vertex in $L$, which must be $v_2$.  It is connected by an edge $\{v_2,v_k\}$ to a vertex $v_k\in R$.  If $k\neq m-1$ (so that $v_k$ is not the rightmost vertex on the path before the leaf $y$), the vertex $v_{k+1}$ is also in $R$ and is connected to an edge $\{v_\ell,v_{k+1}\}$ for some $v_\ell\in L$.  Perform the NNI
\[v_2,v_k,v_{k+1},v_\ell \ \to\ v_2,v_{k+1},v_k,v_\ell,
\]
so that $v_2$ is now connected by an edge to a vertex one further to the right. 
This is again tree-based by Lemma~\ref{l:tree-based-preserving-moves}(ii). This can be continued until there is an edge $\{v_2,v_{m-1}\}$, at which point the network is shoat.
\end{proof}

\begin{lem}\label{l:shoats.connected}
Shoat networks are a connected subspace of $\TBN(n,k)$.  That is,
we can use NNI moves between tree-based networks to move from one shoat network $S_1$ on leaf set $X$ with tier $k$ to another such shoat network $S_2$.
\end{lem}

\begin{proof}
Note that both $S_1$ and $S_2$ have two distinguished leaves (corresponding to $x$ and $y$ in Figure \ref{f:shoat}), and let $x_1$ and $y_1$ be the distinguished leaves in $S_1$, with $x_2$ and $y_2$ the two distinguished leaves in $S_2$. 
We will show how to transform $S_1$ into $S_2$.

If the pairs $\{x_1,y_1\}$ and $\{x_2,y_2\}$ coincide, we do not have to bring the distinguished leaves into their positions: they already are. However, if, say, $x_1$ is not a distinguished leaf in $S_2$, we have to use NNI moves to bring it to its correct position according to $S_2$ (left of the ``left'' vertices), and we do this along the maximal length path from $x_1$ to $y_1$ in $S_1$. Note that all these NNI moves are necessarily legal because $x_1$ is moved left past all the ``left'' vertices in $S_1$, meaning at no point the main rule for NNI (that the first and third or second and fourth vertices are adjacent) is violated.  They also  preserve tree-basedness because they are operations of the first type in Lemma \ref{l:tree-based-preserving-moves}. 

Once we are done with bringing the first distinguished leaf of $S_1$ which is not a distinguished leaf of $S_2$ into position, we bring the corresponding distinguished leaf of $S_2$, that is misplaced in $S_1$, to the former position of $x_1$ by NNI moves. Again, these NNI moves are necessarily legal and preserve tree-basedness by similar arguments to the above.

We repeat the same procedure with $y_1$ if applicable.

When both distinguished leaves are in the correct position (as induced by $S_2$), we use NNI moves to arrange the other $n-2$ leaves correctly.  Again note that all these NNI moves are necessarily legal and preserve tree-basedness by Lemma \ref{l:tree-based-preserving-moves}.

Last, we re-arrange the at most $k-1$ ``left'' vertices between $x_2$ and the non-distinguished leaves of $S_2$ as well as the at most $k-1$ ``right'' vertices between $y_2$ and the non-distinguished leaves of $S_2$ by NNI moves (it is sufficient to rearrange just either the left or the right sides). Again note that all these NNI moves are necessarily legal and preserve tree-basedness by Lemma \ref{l:tree-based-preserving-moves}.
\end{proof}

\section{Connectedness of $\TBN(n,k)$}\label{s:TBN.connectedness}

We are now in a position to prove the main result of this manuscript, namely Theorem~\ref{t:TBN.NNI.connected}.
The proof of this is informed by the idea of the proof of \cite[Theorem 3]{francis2018bounds}, which describes a particular path through the space of networks. This path is essentially laid out by the lemmas in Section~\ref{s:NNI.among.TBNs}. 

\begin{thm}\label{t:TBN.NNI.connected}
The space of tree-based unrooted phylogenetic networks on leaf set $X$ of tier $k$ is connected under NNI moves.
\end{thm}

\begin{proof} 
The case $k=0$, in which networks are trees, has been established in~\cite{robinson1971comparison}, so we may assume $k\ge 1$ (in fact, since all such networks are tree-based for $k\le 4$~\cite{francis2018tree}, the result is also immediate for $k\le 4$ because the space of networks is connected~\cite{huber2016transforming}).

Lemma~\ref{l:simple.to.lineal} shows that a tree-based network can be transformed into a lineal network while remaining inside $\TBN(n,k)$.  
Lemma~\ref{l:lineal.to.pseudo.H} shows that a lineal tree-based network can be transformed into a shoat network, and Lemma~\ref{l:shoats.connected} shows that any two shoat networks are connected by NNI moves within $\TBN(n,k)$ (both these results are restricted to $k\ge 1$ because shoat networks are only defined for $k\ge 1$). Noting that all NNI moves can always be reversed, this proves the theorem (because you can go from any tree-based network $N$ to any other tree-based network $N'$ by first modifying $N$ into a shoat network, possibly modifying this shoat network to give another shoat network, and then continuing from this second shoat network to $N'$ -- all the time only using NNI moves through tree-based networks).
\end{proof}

In the following we use the fact that the number of vertices in an unrooted phylogenetic network of tier $k$ is exactly\footnote{\label{vertices}{Note that a binary phylogenetic tree with $n$ leaves has precisely $2n-2$ vertices \cite{Semple2003}. Any of the extra $k$ edges added to such a tree may induce two new vertices. 
In total, this is $2n-2+2k = 2(n+k-1)$ vertices.}}  
$2(n+k-1)$.

\begin{cor}
The diameter of the space of tree-based networks {with $n$ leaves and tier $k$} under NNI is at most 
\[f(n,k)=\binom{k-1}{2}+\binom{n-2}{2}+4nk+4k^2+4n-4k-17.\] 
That is, the diameter is in $\mathcal O((n+k)^2)$, and for fixed tier $k$, it is in $\mathcal O(n^2)$.
\end{cor}

\begin{proof}
Transforming one tree-based network to another takes the steps described in Theorem~\ref{t:TBN.NNI.connected}: 
\begin{enumerate}
    \item make the networks lineal; 
    \item make the lineal networks shoat graphs; and 
    \item transform between shoat graphs.
\end{enumerate}
Note that as the NNI moves are reversible, the total pathlength between two networks will be at most twice the sum of the maxima of the first two steps, plus the maximum of the third.

The algorithm for making a network lineal, given in Lemma~\ref{l:simple.to.lineal}, takes a support tree with a maximal path, and incorporates other edges of the tree into that path while remaining tree-based.  Each edge is incorporated in a single NNI move, so the number of moves required to make the network lineal is at most the number of edges in a spanning tree minus the number of leaves (because leaf edges are not incorporated into the path), minus the length of the initial maximal path.

We need to bound two numbers: the number of edges in a support tree, and the length of a maximal path.  The number of edges in a support tree is the number of vertices minus one, and for a binary network the number of vertices\footnote{ibid.}
is $2(n+k-1)$. 
Therefore, the number of edges in the support tree is at most $2(n+k-1)-1$.  The number of edges in a maximal path in a support tree would be 2 if the network was a star tree, so we can use this as a very weak lower bound on this.

That is, there are at most $\left(2(n+k-1)-1\right)-n-2=n+2k-5$ moves required.
\[\text{network $\to$ lineal: }n+2k-5.\]

To transform a lineal network into a shoat graph using the procedure described in Lemma~\ref{l:lineal.to.pseudo.H}, we need to shift the ``left'' vertices to the left, and the vertices attached to leaves, to the middle.  Beginning with the ``left'', each must pass by at most $n-2$ vertices attached to leaves, and $2k-2$ other vertices attached to edges (since there must be at least one ``right'' vertex to its right, and one other, since it cannot be adjacent to its partner.  This gives a total of at most $k(n-2+2k-2)=k(n+2k-4)$ moves.  
To move the $n-2$ leaf-adjacent vertices to the middle, they each have to pass at most $k$ left or $k$ right vertices, for a total of $k(n-2)$ moves.
Thus there is a total of $k(2n+2k-6)$ moves to make the network shoat. 
\[\text{lineal $\to$ shoat: }k(2n+2k-6).\]

Finally, a shoat network $N_1$ with distinguished leaves $x_1$ and $y_1$ can be arranged into another (along the lines of Theorem~\ref{t:TBN.NNI.connected}), in $\binom{k-1}{2}+\binom{n-2}{2}+2n+4k-7$ NNI moves: 
\begin{itemize}
    \item If one or both distinguished leaves are distinguished in both shoat networks, then the symmetry of the graph allows us to flip $N_1$ in the vertical axis so that these are aligned.  Thus either one or both distinguished leaves will need to be moved into position from the set of leaves in the middle (between the vertices in $L$ and $R$). This will take $(n-1)+(k-1)$ for the first (moving past possibly $n-1$ leaves and $k-1$ vertices in say $L$), and $(n-2)+(k-1)$ for the second, since there is one fewer leaf available to move past.  The leaves that were previously distinguished then move back to the middle in $k-1$ moves each, for a total of $(n-1)+(k-1)+(n-2)+(k-1)+2(k-1)=2n+4k-7$ moves;
    \item at most $\binom{n-2}{2}$ moves to arrange the leaves in the leaf zone correctly (the diameter of the symmetric group on $n-2$ objects, under the operation of adjacent transpositions); and
    \item at most $\binom{k-1}{2}$ moves to arrange the $k-1$ endpoints on the right or left to agree with $N_2$.
\end{itemize} 
\[\text{shoat $\to$ shoat: }\binom{k-1}{2}+\binom{n-2}{2}+2n+4k-7.\]

Therefore the total maximal distance between any two networks is at most
\[\begin{split}\binom{k-1}{2}+\binom{n-2}{2}+(2n+4k-7)+2\left( (n+2k-5) + k(2n+2k-6)\right)\\ 
=
\binom{k-1}{2}+\binom{n-2}{2}+4nk+4k^2+4n-4k-17,
\end{split}
\]
as required. This completes the proof.
\end{proof}

We are now able to show the connectivity of the entire space of tree-based networks on a fixed set of leaves $X$,  using the triangle operations defined in Section~\ref{s:definitions}. 
First, it is straightforward to show that either triangle operation on a tree-based network is still tree-based, as follows: 

\begin{lem}\label{l:triangle.moves}
If $N$ is a tree-based network in $\TBN(n,k)$, then
\begin{enumerate}
    \item for $v\in V(N)\setminus X$, $\Delta^+(N,v)\in\TBN(n,k+1)$; and
    \item for $\{v_1,v_2,v_3\}$ a 3-cycle in $N$, $\Delta^-(N,\{v_1,v_2,v_3\})\in\TBN(n,k-1)$.
\end{enumerate}
\end{lem}
\begin{proof}

If $N$ is tree-based and $v\in V(N)$ has degree 3, then any support tree for $N$ includes $v$ and at least two of the three edges incident to it, $\{v,w_1\}$, $\{v,w_2\}$, $\{v,w_3\}$, because the support tree is a spanning tree of $N$ whose only leaves are elements of $X$, and hence not $v$.  

{We first show (1)}. Without loss of generality, suppose there is a support tree $T$ for $N$ containing $\{v,w_1\}$ and $\{v,w_2\}$. Then replacing the edges $\{v,w_1\}$ and $\{v,w_2\}$ in $T$ by the edges $\{w_1,v_1\}, \{v_1,v_3\}, \{v_3,v_2\}, \{v_2,w_2\}$ creates a tree (no cycles are generated) that is a support tree for $\Delta^+(N,v)$ (it is connected and has no additional leaves), which is also now in tier $k+1$, {so in summary, we have $\Delta^+(N,v)\in\TBN(n,k+1)$}.

Likewise, {for (2),} if $N$ is tree-based and has a 3-cycle $\{v_1,v_2,v_3\}$, with additional edges $\{v_1,w_1\}, \{v_2,w_2\}, \{v_3,w_3\}$, then it has a support tree that includes all of $\{v_1,v_2,v_3\}$, and at least two of the edges $\{v_1,w_1\}, \{v_2,w_2\}, \{v_3,w_3\}$.  Replacing these by a single vertex in both the network and the support tree does not disconnect the tree, create cycles, or make any vertex uncovered by the collapsed tree. 
\end{proof}

\begin{thm}\label{t:all.TBN.connected}
The space of tree-based networks on $X$ is connected under NNI moves and triangle moves.
\end{thm}

\begin{proof}
This is immediate from Theorem~\ref{t:TBN.NNI.connected} and Lemma~\ref{l:triangle.moves} as long as we can show that it is always possible to move to a tree-based network with a 3-cycle in it, in order to move down a tier.

If the network does not have a triangle already, we can produce one using a series of NNI moves that stays within $\TBN(n,k)$ as follows. 

Let $T$ be a support tree for $N$, and find a cycle in $N$ that is made up of edges in $T$ except for one. 
Such a cycle always exists, just by considering an edge $\{v_1,v_2\}$ in $N\setminus T$ and a path within $T$ between its endpoints.  Denote the vertices of this path $(v_3,\dots,v_i)$, with $i>3$, so that the cycle in $N$ is $(v_1,v_2,\dots,v_i,v_1)$.  Since all vertices in the cycle have degree three in $N$, there is an edge $\{v_3,v_3'\}$ that is not in the cycle.  

Perform an NNI on the path $(v_1,v_2,v_3,v_3')$, which is legal because the edges $\{v_1,v_3\}$ and $\{v_2,v_3'\}$ are not in $N$ (if they were $N$ would have a triangle).  The result is still a tree-based network because if $\{v_3,v_3'\}$ is not in $T$, then Lemma~\ref{l:tree-based-preserving-moves}(ii) applies, whereas if $\{v_3,v_3'\}\in T$, then Lemma~\ref{l:tree-based-preserving-moves}(iii) applies.
This network has a cycle $(v_1,v_3, v_4,\dots,v_i,v_1)$, which is of strictly shorter length {than the original cycle (because it does not contain $v_2$ anymore)}.

In this way we can use NNI to shorten any cycle in $N$ of length greater than 3, and so eventually obtain a triangle.  Deleting that triangle via a triangle move remains tree-based, which completes the proof.
\end{proof}

We end by noting that while the diameter of the full space of tree-based networks on $n$ leaves is infinite, since the tier is unbounded, the distance between a pair of networks in different tiers is also $\mathcal O(n^2)$.

\begin{cor}
The distance between networks $N_1\in\TBN(n,k_1)$ and $N_2\in\TBN(n,k_2)$, for fixed $k_1$ and $k_2$ such that $k_1\le k_2$, is in $\mathcal O(n^2)$.
\end{cor}
\begin{proof}
Given $N_1\in\TBN(n,k_1)$, it can be transformed into a network in $\TBN(n,k_2)$ by successively inserting triangles at vertices $k_2-k_1$ times.  Once in tier $k_2$, it is at most $f(n,k_2)$ moves from $N_2$, and so the distance between $N_1$ and $N_2$ is at most $f(n,k_2)+(k_2-k_1)$.  Since $f(n,k_2)$ is in $\mathcal O(n^2)$ for fixed $k_2$, the result follows.
\end{proof}
 
\section{Discussion}\label{s:discussion}

In this paper we have shown that the space of unrooted binary tree-based networks within a given tier is connected under NNI moves, and that if the triangle moves are included the whole space of unrooted binary tree-based networks is connected (Theorems~\ref{t:TBN.NNI.connected} and~\ref{t:all.TBN.connected}).
This connectivity gives rise to several possible applications and interesting additional questions.

The connectivity of the space means that the class of tree-based networks can be searched by use of NNI moves. This has potential benefits for the sampling of the space using MCMC approaches, as it already does for {phylogenetic trees \cite{Lakner_2008,Hoehna2011,Whidden_2015} and in future most likely also will for general phylogenetic networks \cite{husonrupp}}. 
To further develop this application, it would be interesting to obtain estimates of the size of the tree-based neighbourhood of a given network in $\TBN(n,k)$.

The alert reader may note that the upper bound on the diameter of the space of all tier-$k$ unrooted phylogenetic networks $N(n,k)$ obtained in~\cite{francis2018bounds} is of a similar order to the upper bound on the diameter for the space of tier-$k$ tree-based networks $\TBN(n,k)$ obtained here.  While the density of $\TBN(n,k)$ in $N(n,k)$ is 1 for $n\le 3$, because the sets are equal~\cite{fischer2018nonbinary}, it is likely that the density decreases markedly as $n$ grows. That said, this is unknown and would be interesting to establish. The similarity in the diameters of the spaces could be due to the fact that two tree-based networks might be close in $N(n,k)$, but distant in $\TBN(n,k)$, because the distance in the latter requires a path through networks that remain tree-based.  Or it could be due to the subspace of tree-based networks making up the majority of the space of networks, which seems unlikely but has not been ruled out.  Finally it could be simply that there are tree-based networks at the ``extremities'' of the space of all networks, and it just happens that they are dispersed widely within the space.

There are questions relating to the proximity measures introduced in~\cite{fischer2019how} (generalizing those for rooted networks in~\cite{francis2018new}).  In that paper, networks are given ranks according to their distance from the boundary of the space of tree-based networks.  The boundary of $\TBN(n,k)$ is the subset of tree-based networks that can be made non-tree-based with a single NNI move.  Networks in the boundary are said to have \emph{tree-based rank} 0.  Tree-based networks that are at most $i$ NNI moves from the boundary are rank $i$, and non-tree-based networks that are at most $i$ NNI moves from the boundary are rank $-i$.  {For instance, the two non-tree-based tier-5 networks in Figure \ref{two_non_tb_networks} (taken from~\cite[Figure~7]{fischer2018nonbinary}) are rank $-1$.}

\begin{figure}[ht]
\includegraphics[width=12cm]{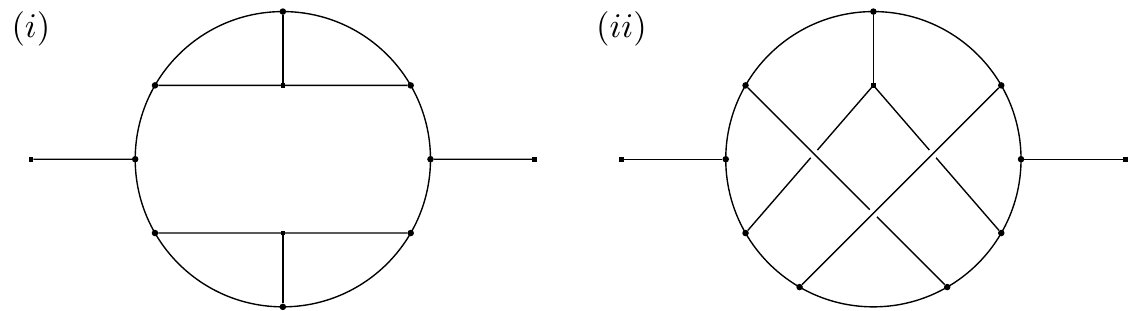}
\caption{
 The only two binary, non-tree-based phylogenetic networks of tier~\cite[Figure~7]{fischer2018nonbinary}. These are both rank $-1$, that is, for both of them it takes only one NNI move to make them tree-based. However, no single NNI move can convert one of these networks into the other one. 
 }\label{two_non_tb_networks}
\end{figure}

These two are not connected by a single NNI move\footnote{We used the computer algebra system Mathematica \cite{Mathematica} to verify that the shortest path from network (i) in Figure \ref{two_non_tb_networks} to a network isomorphic to (ii) requires 5 NNI moves. While this is tricky to see, it is combinatorially rather easy to see that network (i) has 24 1-step NNI neighbours, which can be divided into two classes {(i.e. the NNI neighbourhood of network (i) contains only two non-isomorphic networks)}: {those isomorphic to (i), and those isomorphic to a specific different network, which is} in fact tree-based. So network (ii) cannot be in the 1-step neighborhood of (i).}, so the set of networks of tree-based rank $-1$ is not connected.  Is this also the case for other negative ranks?  And especially, what about networks of positive tree-based rank, in the ``core'', as it were, of the space of tree-based networks?
{While it is certainly of interest to use the results of ~\cite{fischer2019how} to analyze the structure of the space of binary tree-based networks more in-depth, note that in~\cite{fischer2019how}, the authors discuss general tree-based networks, that is, non-binary ones. On the other hand, the results of the present manuscript are limited to the binary case. This is not merely due to the fact that binary networks play a fundamental role in mathematical phylogenetics. On the contrary: It can be easily seen that the space of \emph{non-binary} tree-based networks is \emph{not} connected under NNI. To see this, note that NNI does not change the total degree of the underlying graph, which is why an NNI move cannot connect networks whose total degrees differ.}

A direct conclusion from the present manuscript is the fact that Theorem~\ref{t:TBN.NNI.connected} shows that $\TBN(n,k)$ is also connected under the related SPR and TBR moves, which have recently been defined for phylogenetic networks, because all NNI moves are also SPR and TBR moves~\cite[Lemma 7]{francis2018bounds}. But we have left open the question of the diameter of the space under these moves, except to note that clearly it is also within $\mathcal O(n^2)$, and so may be smaller.

{However, there are two more open questions arising from this manuscript: First, our approach of using NNI to go from a binary unrooted network to a lineal network, then to a shoat network and then modify the shoat network, go back to a lineal network and finally arrive at the target network does probably not lead to a shortest path. The approach is constructive as it shows how to construct a valid tree-based path from any binary tree-based network $N$ to any other one $N_2$, but there may be shorter paths from $N_1$ to $N_2$. In particular, it might be possible to avoid shoat networks altogether and directly move from one lineal network to another one. If this was possible, it would reduce the complexity of our construction.}

Second, it should be noted that NNI moves on phylogenetic networks have been introduced in the \emph{rooted} setting~\cite{gambette2017rearrangement}, and of course this is also the original setting for the introduction of tree-based phylogenetic networks~\cite{francis2015phylogenetic}.  While there are concrete connections between rooted and unrooted tree-based networks (eg~\cite[Section~6]{fischer2019how}), it is not clear whether the results in the present paper will lift to the rooted context.

\section{Acknowledgements} 
MF wishes to thank the DAAD for conference travel funding to the Annual New Zealand Phylogenomics Meeting, where partial results leading to this manuscript were achieved.

\bibliographystyle{plain}

\begin{thebibliography}{10}
 
\bibitem{allen2001subtree}
Benjamin~L Allen and Mike Steel.
\newblock Subtree transfer operations and their induced metrics on evolutionary
  trees.
\newblock {\em Annals of combinatorics}, 5(1):1--15, 2001.

\bibitem{fischer2019how}
Mareike Fischer and Andrew Francis.
\newblock How tree-based is my network? {Proximity} measures for unrooted
  phylogenetic networks.
\newblock {\em arXiv:1906.06163}, 2019.

\bibitem{fischer2018classes}
Mareike Fischer, Michelle Galla, Lina Herbst, Yangjing Long, and Kristina
  Wicke.
\newblock Classes of treebased networks.
\newblock {\em arXiv:1810.06844}, 2018.

\bibitem{fischer2018nonbinary}
Mareike Fischer, Michelle Galla, Lina Herbst, Yangjing Long, and Kristina
  Wicke.
\newblock Non-binary treebased unrooted phylogenetic networks and their
  relations to binary and rooted ones.
\newblock {\em arXiv:1810.06853}, 2018.

\bibitem{francis2018tree}
Andrew Francis, Katharina~T Huber, and Vincent Moulton.
\newblock Tree-based unrooted phylogenetic networks.
\newblock {\em Bulletin of Mathematical Biology}, 80(2):404--416, 2018.

\bibitem{francis2018bounds}
Andrew Francis, Katharina~T Huber, Vincent Moulton, and Taoyang Wu.
\newblock Bounds for phylogenetic network space metrics.
\newblock {\em Journal of Mathematical Biology}, 76(5):1229--1248, 2018.

\bibitem{francis2018new}
Andrew Francis, Charles Semple, and Mike Steel.
\newblock New characterisations of tree-based networks and proximity measures.
\newblock {\em Advances in Applied Mathematics}, 93:93--107, February 2018.

\bibitem{francis2015phylogenetic}
Andrew~R Francis and Mike Steel.
\newblock Which phylogenetic networks are merely trees with additional arcs?
\newblock {\em Systematic Biology}, 64(5):768--777, 2015.

\bibitem{gambette2012quartets}
Philippe Gambette, Vincent Berry, and Christophe Paul.
\newblock Quartets and unrooted phylogenetic networks.
\newblock {\em Journal of bioinformatics and computational biology},
  10(04):1250004, 2012.

\bibitem{gambette2017rearrangement}
Philippe Gambette, Leo van Iersel, Mark Jones, Manuel Lafond, Fabio Pardi, and
  Celine Scornavacca.
\newblock Rearrangement moves on rooted phylogenetic networks.
\newblock {\em PLOS Computational Biology}, 13(8):e1005611, August 2017.

\bibitem{hendriksen2018tree}
Michael Hendriksen.
\newblock Tree-based unrooted nonbinary phylogenetic networks.
\newblock {\em Mathematical Biosciences}, 302:131--138, 2018.

\bibitem{Hoehna2011}
Sebastian H\"ohna and Alexei~J. Drummond.
\newblock {Guided Tree Topology Proposals for Bayesian Phylogenetic Inference}.
\newblock {\em Systematic Biology}, 61(1):1--11, 08 2011.

\bibitem{huber2016transforming}
Katharina~T Huber, Vincent Moulton, and Taoyang Wu.
\newblock Transforming phylogenetic networks: Moving beyond tree space.
\newblock {\em Journal of Theoretical Biology}, 404:30--39, 2016.

\bibitem{husonrupp}
Daniel~H. Huson, Regula Rupp, and Celine Scornavacca.
\newblock {\em Phylogenetic Networks: Concepts, Algorithms and Applications}.
\newblock Cambridge University Press, 2010.

\bibitem{Mathematica}
Wolfram~Research{,} Inc.
\newblock Mathematica, {V}ersion 10.3, 2017.
\newblock Champaign, IL, 2017.

\bibitem{janssen2019rearrangement}
Remie Janssen and Jonathan Klawitter.
\newblock Rearrangement operations on unrooted phylogenetic networks.
\newblock {\em arXiv:1906.04468}, 2019.

\bibitem{jetten2016nonbinary}
Laura Jetten and Leo van Iersel.
\newblock Nonbinary tree-based phylogenetic networks.
\newblock {\em IEEE/ACM Transactions on Computational Biology and
  Bioinformatics}, 2016.

\bibitem{Lakner_2008}
Clemens Lakner, Paul van~der Mark, John~P. Huelsenbeck, Bret Larget, and
  Fredrik Ronquist.
\newblock {Efficiency of Markov Chain Monte Carlo Tree Proposals in Bayesian
  Phylogenetics}.
\newblock {\em Systematic Biology}, 57(1):86--103, 02 2008.

\bibitem{pons2018tree}
Joan~Carles Pons, Charles Semple, and Mike Steel.
\newblock Tree-based networks: characterisations, metrics, and support trees.
\newblock {\em Journal of Mathematical Biology}, pages 1--20, 2018.

\bibitem{robinson1971comparison}
David~F. Robinson.
\newblock Comparison of labeled trees with valency three.
\newblock {\em Journal of Combinatorial Theory, Series B}, 11(2):105--119,
  1971.

\bibitem{Semple2003}
Charles Semple and Mike Steel.
\newblock {\em Phylogenetics ({O}xford {L}ecture {S}eries in {M}athematics and
  {I}ts {A}pplications)}.
\newblock Oxford University Press, 2003.

\bibitem{steel2016phylogeny}
Mike Steel.
\newblock {\em Phylogeny: discrete and random processes in evolution}.
\newblock SIAM, 2016.

\bibitem{Whidden_2015}
Chris Whidden and A.~Frederick~Matsen~IV.
\newblock {Quantifying MCMC Exploration of Phylogenetic Tree Space}.
\newblock {\em Systematic Biology}, 64(3):472--491, 01 2015.

\end{thebibliography}

\end{document}